\documentclass[journal,10pt]{IEEEtran}
\usepackage[T1]{fontenc}
\usepackage[latin1]{inputenc}
\usepackage{fancybox,shadow,xcolor} 
\usepackage{setspace}
\usepackage{placeins}

\usepackage{psfrag}
\usepackage{bm}
\usepackage{float}
\usepackage{subfig}

\usepackage{graphicx}
\usepackage{amssymb,amsmath,amsthm}
\usepackage{mathrsfs}

\usepackage{algorithmic,algorithm}

\newtheorem{thm}{Theorem}

\newtheorem{prop}{Proposition}
\newtheorem{definition}{Definition}
\newtheorem{cor}{Corollary}
\newtheorem{rem}{Remark}

\DeclareMathOperator*{\argmin}{arg\,min}

\DeclareMathOperator*{\rank}{rank}

\DeclareMathOperator*{\im}{im}

\DeclareMathOperator*{\F}{F}

\renewcommand{\Re}{\mathbb{R}}

\usepackage{bbm}
\renewcommand{\paragraph}[1]{\smallskip\noindent\textbf{#1.} }

\newcommand{\BM}{\begin{bmatrix}}
\newcommand{\EM}{\end{bmatrix}}

\newcommand{\BBM}{\big[\begin{matrix}}
\newcommand{\EEM}{\end{matrix}\big]}

\newcommand{\bbm}{[\begin{matrix}}
\newcommand{\eem}{\end{matrix}]}

\title{On a class of optimization-based robust estimators}
\author{Laurent Bako
\thanks{L. Bako is with Laboratoire Amp\`{e}re -- Ecole Centrale de Lyon -- Universit\'{e} de Lyon -- 36 avenue Guy de Collongue, 69134 Ecully, France -- {\tt E-mail: laurent.bako@ec-lyon.fr}}
}

\begin{document}
\maketitle

\begin{abstract}
We consider in this paper the problem of estimating a parameter matrix from observations which are affected by two types of noise components: (i) a sparse noise sequence which, whenever nonzero can have arbitrarily large amplitude (ii) and a dense and bounded noise sequence of "moderate"  amount. This is termed a robust regression problem. To tackle it, a quite general optimization-based  framework is proposed and analyzed.  When only the sparse noise is present, a sufficient  bound is derived on the number  of nonzero elements in the sparse noise sequence that can be accommodated by the estimator while still returning the true parameter matrix. While almost all the  restricted isometry-based bounds from the literature are not verifiable, our bound can be easily computed through solving a convex optimization problem. Moreover, empirical evidence tends to suggest that it is generally tight. If in addition to the sparse noise sequence, the training data are affected by a bounded dense noise, we derive an upper bound on the estimation error. 
\end{abstract}



\section{Introduction}
In many engineering fields such as control system design, signal processing, machine learning or statistics, one is frequently confronted with the problem of empirically uncovering a mathematical relationship between a number of signals of interest. The usual method to achieve this goal is to run an experiment during which one measures (a finite number of) samples of the relevant signals and proceed with fitting a certain model structure to the experimental data samples. This process is known as system identification \cite{LjungBook,Soderstrom-Book89}.  A issue of critical importance during this process is that the experimental data samples might be contaminated by a measurement noise of relatively high level due for example to intermittent sensor failures or various communication disruptions. To cope with the troublesome effects of the noise, the model estimation must be designed with care. \\
In this paper we consider the situation where the data are corrupted by two types of noise: a sparse noise sequence which shows up only intermittently in time but can take on arbitrarily large values whenever it is nonzero; and a more standard dense noise component of moderate amount.   
 

\section{The robust regression problem}
Consider a system described by an equation of the form
\begin{equation}\label{eq:model-for-robust-estimation}
	y_t=A^o x_t +f_t+e_t
\end{equation}
where $y_t\in \Re^m$ and  $x_t\in \Re^n$ are respectively the output and the regressor vector at time $t$; $A^o\in \Re^{m\times n}$ is an unknown parameter matrix; $f_t$ and $e_t$ are some noise terms which are  unobserved. 

\paragraph{Problem}
Given a finite collection $\left\{x_t,y_t\right\}_{t=1}^N$ of measurements obeying the relation \eqref{eq:model-for-robust-estimation}, the robust regression problem of interest here is the one of finding an estimate of the parameter matrix $A^o$ under the assumptions that $\left\{e_t\right\}$ and $\left\{f_t\right\}$ are unknown but enjoy the following (informal) properties: 
\begin{itemize}
	\item $\left\{e_t\right\}$ is a  dense noise sequence with bounded elements accounting for moderate model mismatches or measurement noise. 
	\item $\left\{f_t\right\}$ is such that the majority of its elements are equal to zero while the remaining  nonzero elements can be of arbitrarily large magnitude. The nonzero elements of that sequence are usually termed gross errors or outliers. They can  account for possible intermittent sensor faults. We will refer to $\left\{f_t\right\}$ as the sequence of sparse noise.   
\end{itemize}
\noindent For the time being, these are just informal descriptions of the characteristics of the sequences $\left\{f_t\right\}$ and $\left\{e_t\right\}$. They will be made more precise whenever necessary in the sequel for the need of stating more formal results. 


\noindent Let $Y\in \Re^{m\times N}$ and $X\in \Re^{n\times N}$ be data matrices formed respectively with $N$ output measurements and regressor vectors. Then it follows from \eqref{eq:model-for-robust-estimation} that 
\begin{equation}\label{eq:model}
	Y=A^oX+E+F, 
\end{equation}
where  $E\in \Re^{m\times N}$ and $F\in \Re^{m\times N}$ are unknown noise components. The matrices $Y$ and $X$ can be structured or not,  depending on whether the system \eqref{eq:model-for-robust-estimation} is dynamic or not. For example, when the model \eqref{eq:model-for-robust-estimation} is of   MIMO FIR type, $Y$ contains a finite collection of output measurements while $X$ is a Hankel matrix containing  lagged inputs of the system. In this case $Y$ and $X$ take the form 
$$
\begin{aligned}
	&Y = \BM y_1 & y_2& \cdots & y_N \EM, \\
	& X = \BM u_1 & u_2 & \cdots& u_N\\
	u_0 & u_1 & \cdots & u_{N-1}\\
	\vdots & \vdots & \cdots & \vdots\\
	u_{1-n_f} & u_{2-n_f} & \cdots & u_{N-n_f}\EM.
\end{aligned}
$$ 
where $\left\{u_t\right\}$ and $\left\{y_t\right\}$ stand respectively for the input and output of the system and  the maximum lag $n_f$ is called the order of the model. In the sequel, the notations of the type $y_t$ and $x_t$ with subindex $t\in \mathbb{I}\triangleq\left\{1,\ldots,N\right\}$ refer to the columns of $Y$ and $X$ respectively.


\paragraph{Relevant prior works}
The so formulated regression problem is called a robust regression problem in connection with the fact that the error matrix $F$ assume columns of (possibly) arbitrarily large amplitude. It has applications in e.g., the identification of switched linear systems \cite{Bako11-Automatica,Ozay12-TAC,Ozay11-CDC-ECC}, subspace clustering \cite{Bako14-SPL}, etc. 
Existing approaches for solving the robust regression problem can be roughly divided into two groups: methods from the field of robust statistics \cite{Rousseeuw05-Book,Maronna06-Book,Huber-Book-09} which have been developed since the early 60s and a class of more recent methods inspired by the compressed sensing paradigm \cite{Bako16-Automatica,Candes06-IT,Sharon09-ACC,Xu14-Automatica,Mitra13-SP}. 
The first group comprises methods such as the least absolute deviation (LAD) estimator \cite{Huber87-L1}, the least median of squares \cite{Rousseeuw84}, the least trimmed squares \cite{Rousseeuw05-Book}, the family of M-estimators \cite{Huber-Book-09}.  The latter group can be viewed essentially as a refreshed look at the so-called least absolute deviation method. There has been however a fundamental shift of philosophy in  the analysis. While in the framework of robust statistics, robustness of an estimator is measured in terms of the breakdown point (the asymptotic minimum proportion of points which cause the estimation error induced by an 
estimator to be unbounded if they were to be arbitrarily corrupted
by gross errors), in the compressed-sensing-inspired category of robust methods, the analysis aims generally at characterizing properties of the data that favor exact recovery of the true parameter matrix $A^o$.  In this latter group, the LAD estimator is sometimes regarded as a convex relaxation of a combinatorial sparse optimization problem.

To the best of our knowledge, only the papers \cite{Sharon09-ACC} provides an explicit bound on the estimation error induced by the LAD estimator. However that bound does not fully apply to the current setting since the estimators although similar are of different natures. Indeed, the LAD  estimator stands only as a special case of the current framework. Moreover the bound in \cite{Sharon09-ACC} is not easily computable while ours is. The references \cite{Candes06-IT} and \cite{Mitra13-SP} provide some bounds for a noise-aware version of the LAD estimator which are based respectively on the Restricted Isometry Property (RIP) and a measure of subspace angles. Unfortunately numerical evaluation of those bounds is a process of exponential complexity, a price that is unaffordable in practice. 

A related but  different problem from the regression problem considered here is that of sparse signal recovery studied in the field of compressed sensing \cite{Candes08-SPM,Foucart13-Book}. This is about finding the sparsest solution to an underdetermined set of linear equations. Various analysis approaches have been devised which rely on the RIP constant, the mutual coherence, the nullspace property, to name but a few. Again, these analysis results either cannot be extended efficiently to the robust regression problem or lead to bounds that are NP-hard to compute \cite{Tillmann14-IT,Juditsky11-MP,Donoho06-IT}.

\paragraph{Contributions}
In this paper we propose and analyze a class of optimization-based robust estimators. It is shown that the robust properties of the estimators are essentially inherited from a key property of the to-be-optimized performance function (or loss function) called column-wise summability. The proposed framework admits the LAD estimator and its usual variants as special cases. Moreover it applies to both SISO and MIMO systems.  
When the dense noise component $E$ in \eqref{eq:model} is identically equal to zero, we derive  bounds on the number of gross errors (nonzero columns of $F$) that the estimator is able to accommodate while still returning the true parameter matrix $A^o$. In comparison with the existing literature, the proposed bounds have the important advantage that they are numerically computable through convex optimization. When both $E$ and $F$ are active,  exact recovery of the true parameter matrix is no longer possible. 
In this scenario, we derive upper bounds on the parametric estimation error in function of the amplitude of $E$ and the number of nonzero columns of $F$. Again, computable but (possibly) looser versions of those bounds are obtainable.  

The current paper can be viewed as a generalization of our previous work reported in \cite{Bako16-Automatica}. While \cite{Bako16-Automatica} provides an analysis of mostly a single estimator (namely the LAD estimator) relying on nonsmooth optimization theory, we focus here on a much larger class of optimization-based robust estimators by highlighting some key robustness-inducing properties. Moreover, we provide, for the considered class of estimators, stability results which permit the estimation of parametric error bounds.  

\paragraph{Outline}The rest of the paper is organized as follows. Section \ref{sec:nonsmooth} defines the optimization-based approach to the robust regression problem. Section \ref{sec:analysis} discusses the properties of the proposed estimation framework.  Section \ref{sec:special-case} provides further comments. 
Section \ref{sec:experiments} reports some numerical experiments. Lastly, Section \ref{sec:conclusion} contains some concluding remarks.

\paragraph{Notations} 
$\Re$ is the set of real numbers;  $\Re_{\geq 0}$ (respectively $\Re_{> 0}$) is the set of nonnegative (respectively positive) real numbers ; $\Re^N$ is the space of $N$-tuples (vectors) of real numbers. 
For any vector $x=[\begin{matrix}x_1 & \cdots & x_N \end{matrix}]^\top\in \Re^N$, the $p$-norm of $x$ with $p\in \left\{1,\ldots,\infty\right\}$  is defined by
$\left\|x\right\|_p=\big(\sum_{i=1}^N \left|x_i\right|^p\big)^{1/p}$. A special case is the limit case $p=\infty$ in which  $\left\|x\right\|_\infty=\max_{i=1,\ldots,N}\left|x_i\right|$. For any matrix $A=[\begin{matrix}a_1 & \cdots & a_N \end{matrix}]$ with $a_i\in \Re^m$, the induced $p$-norm of $A$ is defined by $\left\|A\right\|_p=\sup_{x\in \Re^N,\left\|x\right\|_p=1}\left\|Ax\right\|_p$. \\
\textit{Cardinality of a finite set.} Throughout the paper, whenever $\mathcal{S}$ is a finite set, the notation $\left|\mathcal{S}\right|$ will refer to the cardinality of $\mathcal{S}$. However, for a real number $x$, $\left|x\right|$ will denote the absolute value of $x$.\\
\textit{Submatrices and subvectors.}
Let $X \in \Re^{n\times N}$ and  $\mathbb{I}=\left\{1,\ldots,N\right\}$ be the index set for the columns of $X$. If $I\subset \mathbb{I}$, the notation $X_I$ denotes a matrix in $\Re^{n\times \left|I\right|}$ formed with the columns of $X$ indexed by $I$.
We will use the convention that $X_I=0\in \Re^n$  when the  index set $I$ is empty.

\section{A class of robust estimators}\label{sec:nonsmooth}
Let $\mathcal{D}_N$ be the set of $N$ data points generated by system  \eqref{eq:model-for-robust-estimation} for any possible values of the noise sequences,  i.e., 
$$\begin{aligned}
	\mathcal{D}_N=\Big\{(Y,X)&\in \Re^{m\times N}\times \Re^{n\times N}:\Big.  \\ 
	&\left.\exists (E,F)\in \mathcal{G}_N^e\times \mathcal{G}_N^f, \: \eqref{eq:model} \: \mbox{holds}\right\},
\end{aligned}$$
with $\mathcal{G}_N^e\subset \Re^{m\times N}$ and $\mathcal{G}_N^f\subset \Re^{m\times N}$ denoting the set of dense and sparse noise matrices  respectively. 
The estimation problem aims at  determining the unknown parameter matrix $A^o$ given a point $(Y,X)$ in $\mathcal{D}_N$. Of course, this quest would not make much sense if the noises $E$ and $F$ were completely arbitrary since in this case, we would have $\mathcal{D}_N=\Re^{m\times N}\times \Re^{n\times N}$ hence losing any informativity concerning the data-generating system. Therefore some minimum constraints need to be put on $E$ and $F$ as informally described above. \\  
With respect to the estimation problem just stated, an estimator is a set-valued map $\Psi:\mathcal{D}_N\rightarrow \mathscr{P}(\Re^{m\times n})$, $(Y,X)\mapsto \Psi(Y,X)$ which is defined from the data space $\mathcal{D}_N$ to the power set $\mathscr{P}(\Re^{m\times n})$ of the parameter space.  For $(Y,X)$ generated by a system of the form \eqref{eq:model-for-robust-estimation}, one would like to design an estimator achieving,  whenever possible, the ideal property that $\Psi(Y,X)=\left\{A^o\right\}$. In default of that ideal situation, a more pragmatic goal is to search for a $\Psi$ so that $A^o\in \Psi(Y,X)$ and $\Psi(Y,X)$ is of small size in some sense despite the troublesome effects of the unknown noise components $E$ and $F$. The design of an optimal estimator requires specifying a performance index (usually called a loss function) which is to be minimized.  

In this paper, we study the properties of the estimator of the parameter matrix $A^o$  in \eqref{eq:model} defined by 
\begin{equation}\label{eq:Psi-YX}
	\Psi(Y,X)= \argmin_{A\in \Re^{m\times n}}\varphi(Y-AX)
\end{equation}
where $\varphi: \mathcal{M}\left(\Re\right)\rightarrow \Re_{\geq 0}$ is a \textit{convex function} defined on the set $\mathcal{M}\left(\Re\right)$  of (all) real matrices. It is assumed that $\varphi$ has the following properties:
\begin{itemize}
	\item[\textbf{P1.}] For all $B, C \in \mathcal{M}\left(\Re\right)$ of compatible dimensions, 
\begin{equation}\label{eq:column-summability}
	\varphi(\bbm B & C\eem)= \varphi(B)+\varphi(C)
\end{equation}
with $\bbm B & C\eem$ denoting the matrix formed by concatenating  column-wise $B$ and $C$. 

\item[\textbf{P2.}] There exists a matrix norm $\ell: \mathcal{M}\left(\Re\right)\rightarrow \Re_{\geq 0}$ such that for all $B,C\in \mathcal{M}\left(\Re\right)$, conformable for addition,
\begin{equation}
	 \varphi(B)\leq \varphi(B-C)+\ell(C)\label{eq:Psi-Ineq1}
\end{equation}
	\item[\textbf{P3.}]There exists a constant real number $\varepsilon\geq 0$  such that for all $B\in \mathcal{M}\left(\Re\right)$ with $n$ rows and $N$ columns, 
\begin{equation}	
\ell(B)-\left|I_\varepsilon^c(B)\right|\varepsilon \leq \varphi(B)\leq \ell(B) \label{eq:Psi-Ineq2}
\end{equation}
where  
$$I_\varepsilon^c(B)= \big\{i\in \left\{1,\ldots,N\right\}: \ell(b_i)>\varepsilon \big\}$$
and $\left|I_\varepsilon^c(B)\right|$ is the cardinality of $I_\varepsilon^c(B)$ and $b_i\in \Re^n$ is the $i$th column of the $(n,N)$-matrix $B$. 
\end{itemize}
\noindent The property \eqref{eq:column-summability} will be  called column-wise summability. Since $\varphi$ is a function defined over the space of real matrices of any dimensions, it is also defined for $n$-dimensional vectors of real numbers. 
Hence according to property \eqref{eq:column-summability}, if $B=\bbm b_1 & \cdots & b_N\eem$ with column vectors $b_i\in \Re^n$, then  
$$\varphi(B)=\sum_{i=1}^N \varphi(b_i). $$
The so-defined function $\varphi$ is not necessarily a norm. For any $\varepsilon^o\geq 0$ and any vector norm $\ell^o$, it can be verified that the function $\varphi$ defined by 
\begin{equation}\label{eq:phi-example}
	\varphi(B)=\sum_{i=1}^N\max(0,\ell^o(b_i)-\varepsilon^o) 
\end{equation}
is  positive and convex and satisfies properties \eqref{eq:column-summability}-\eqref{eq:Psi-Ineq2} but it is not a norm for $\varepsilon^o>0$ since in this case, $\varphi(B)=0$ does not imply that $B=0$. But if $\varepsilon^o=0$ in \eqref{eq:phi-example}, then $\varphi=\ell$ by \eqref{eq:Psi-Ineq2} so that $\varphi$ corresponds to the matrix norm defined by $\varphi(B)=\sum_{i=1}^N\ell^o(b_i)$. We note in this latter case that \eqref{eq:Psi-Ineq2} is trivial while \eqref{eq:Psi-Ineq1} reduces to the triangle inequality.

We will show in the sequel that the estimator $\Psi$ in \eqref{eq:Psi-YX} enjoys some impressive robustness properties with respect to the sparse noise matrix $F$.  The term sparse is used here to mean that a relatively large proportion of the column vectors of $F$ are equal to zero. And saying that $\Psi$ is robust with respect to $F$ means that $\Psi(Y,X)$ does not depend on (or is insensitive to) the magnitudes of the nonzero columns of $F$ under the sparsity condition.  Therefore those few columns which are nonzero  can have arbitrarily large magnitude.
As will be shown in the sequel, the robustness properties of $\Psi$ are inherited from the properties P1-P3 of the objective function $\varphi$. In the special case where $\varphi$ is a norm, the properties P2-P3 are automatically satisfied so that P1 becomes the only key property required.   As to the convexity of $\varphi$, it is intended just for computational reasons as it eases the solving of the optimization problem in \eqref{eq:Psi-YX}.

\section{Properties of the robust estimators}\label{sec:analysis}

\subsection{Exact recoverability}
We first study the conditions under which the true parameter matrix $A^o$ in \eqref{eq:model-for-robust-estimation} can be exactly recovered. Theorem \ref{cor:NS-Conds} and Theorem \ref{thm:sufficient-condition-uniqueness} stated next show that if the number of nonzero columns in the matrix $V\triangleq E+F$ is  less than a certain threshold,  then $\Psi(Y,X)=\left\{A^o\right\}$.

\begin{thm}[A necessary and sufficient condition] \label{cor:NS-Conds}
Let $\varphi$ be a function satisfying  \eqref{eq:column-summability}-\eqref{eq:Psi-Ineq2} with $\varepsilon=0$ and $\Psi$ be defined as in \eqref{eq:Psi-YX}.  
Let $d$ be an integer and assume that $\rank(X)=n$. For any $A\in \Re^{m\times n}$ and $Y\in \Re^{m\times N}$, let $\mathbb{I}^c\left(Y-AX\right)=\left\{t\in \mathbb{I}: y_t-Ax_t\neq 0\right\}$.  Then the following statements are equivalent.
\begin{enumerate}
\itemsep=.05cm
\item[\emph{(i)}]   
\begin{equation}\label{eq:(i)}
\begin{aligned}
		\forall  A\in \Re^{m\times n}, \forall Y\in \Re^{m\times N},& \: \left|\mathbb{I}^c\left(Y-AX\right)\right|\leq  d \\
		                           &\quad \Rightarrow \quad \quad  \Psi(Y,X)=\big\{A\big\}
\end{aligned}
\end{equation}
\item[\emph{(ii)}]
	\begin{equation}\label{eq:(ii)}
\max_{\substack{I^c\subset \mathbb{I}:\\\left|I^c\right|=d}}\: \: \max_{\substack{\Lambda\in \Re^{m\times n}\\\Lambda\neq 0}}\left[\dfrac{\varphi(\Lambda X_{I^c})}{\varphi(\Lambda X)}\right]<\dfrac{1}{2}
\end{equation}
\end{enumerate}
Here and in the following, the notation $\mathbb{I}\triangleq\left\{1,\ldots,N\right\}$ is used to denote the index set for the  columns of the data matrices.
\end{thm}
\begin{proof}
We first note that the rank assumption on $X$ is intended to insure that \eqref{eq:(ii)} is well-defined since then, with $\varphi$ being a norm, $\varphi(\Lambda X)\neq 0$ whenever $\Lambda\neq 0$. \\
(i) $\Rightarrow$ (ii): Assume that (i) holds.\\
Consider an arbitrary subset $I^c$ of $\mathbb{I}$ such that $\left|I^c\right|= d$. Let $\Lambda$\ be any matrix in $\Re^{m\times n}$ satisfying $\Lambda\neq 0$. Finally, consider a matrix $Y\in \Re^{m\times N}$ defined by $Y_{I^c}=0$ and $Y_{I^0}=\Lambda X_{I^0}$ where $I^0=\mathbb{I}\setminus I^c$. Then  $\mathbb{I}^c(Y-\Lambda X)\subset I^c  $ and so $\left|\mathbb{I}^c(Y-\Lambda X)\right|\leq d$. 
Hence by (i) 
$\left\{\Lambda\right\}=\argmin_{H}\varphi(Y-HX)$ which means that $\varphi(Y-\Lambda X)<\varphi(Y-HX)$ for any $H\in \Re^{m\times n}$, $H\neq \Lambda$. In particular, by taking $H=0$ we get $\varphi(Y-\Lambda X)<\varphi(Y)$. It follows from the property \eqref{eq:column-summability} that $$\varphi(Y_{I^c}-\Lambda X_{I^c})+\varphi(Y_{I^0}-\Lambda X_{I^0}) <\varphi(Y_{I^c})+\varphi(Y_{I^0}).$$
Using now the relations $Y_{I^c}=0$ and $Y_{I^0}=\Lambda X_{I^0}$ yields $\varphi(\Lambda X_{I^c})<\varphi(\Lambda X_{I^0})$ or, equivalently, 
$\varphi(\Lambda X_{I^c})  <1/2\varphi(\Lambda X)$. 
Eq. \eqref{eq:(ii)} then follows from the fact that $I^c$ and $\Lambda$ are arbitrary.

\noindent (ii) $\Rightarrow$ (i): 
To begin with, note that if Eq. \eqref{eq:(ii)} holds for some $d$, then it holds also for any $d_0\leq d$. As a result, the equality $|I^c|=d$ in \eqref{eq:(ii)} can be changed to $|I^c|\leq d$.  Assuming (ii), let $A\in \Re^{m\times n}$ and $Y\in \Re^{m\times N}$ be matrices satisfying $\left|\mathbb{I}^c(Y-A X)\right|\leq d$. Set $I^c=\mathbb{I}^c(Y-A X)$ and $I^0=\mathbb{I}\setminus I^c$. Then for all $\Lambda\in \Re^{m\times n}$ such that $\Lambda \neq 0$, 
$$2\varphi(\Lambda X_{I^c})< \varphi(\Lambda X) =\varphi(\Lambda X_{I^c})+ \varphi(\Lambda X_{I^0}),$$
where the equality is obtained by the property \eqref{eq:column-summability} of $\varphi$.
It follows  that
\begin{equation}\label{eq:Inter1}
	\varphi(\Lambda X_{I^c})< \varphi(Y_{I^0}-(A+\Lambda) X_{I^0}). 
\end{equation}
On the other hand, we know by \eqref{eq:Psi-Ineq1} that
$$\begin{aligned}
	\varphi(Y_{I^c}-A X_{I^c})-&\varphi(Y_{I^c}-(A+\Lambda) X_{I^c}) \leq \varphi(\Lambda X_{I^c}).
\end{aligned} $$
Combining with the inequality \eqref{eq:Inter1} yields
$$\varphi(Y-AX)< \varphi(Y-(A+\Lambda)X).$$
Since $\Lambda$ is an arbitrary nonzero matrix, this inequality says that $A$ is the unique minimizer of $V(H)=\varphi(Y-HX)$. 
\end{proof}
\noindent Consider a data pair $(Y,X)$ generated by \eqref{eq:model-for-robust-estimation}.  By letting 
\begin{equation}\label{eq:pic}
	\pi_{\varphi}^c(X)=\max\big\{d: \mbox{Eq.}\: \eqref{eq:(ii)} \mbox{ holds}\big\},  
\end{equation}
and assuming that $\pi_{\varphi}^c(X)>0$ we can see that whenever $\left|\mathbb{I}^c(Y-A^oX)\right|\leq \pi_{\varphi}^c(X)$, $A^o$ can be exactly recovered by computing $\Psi(Y,X)$. Of course this is likely to hold only if the dense noise component $E$ does not exist. So in the situation where $E=0$, the theorem says that $A^o$ can be uniquely obtained by convex optimization provided that the number of outliers (nonzero columns of $F$) is  less than or equal to $\pi_{\varphi}^c(X)$. For the condition of exact recoverability to be checkable we must be able to compute $\pi_{\varphi}^c(X)$. The bad news are that evaluating numerically such a number  is likely to be NP-hard in most cases. \\ 
In the sequel, we investigate sufficient conditions of exact recovery which are more tractable from a numerical standpoint. For this purpose let us introduce some definitions. 
\begin{definition}\label{def:self-decomposable}
A matrix $X=\bbm x_1 & \cdots & x_N\eem\in \Re^{n\times N}$ is said to be self-decomposable if $\rank(X)=n$ and for all $k\in \mathbb{I}$, $x_k\in \im(X_{\neq k})$ where $X_{\neq k}\triangleq X_{\mathbb{I}\setminus \{k\}}$ is the matrix obtained from $X$ by removing its $k$-th column and $\im(\cdot)$  refers to range space. 
\end{definition}
For a matrix to be self-decomposable it is enough that $X_{\neq k}$ be full row rank for any $k\in \mathbb{I}$. Achieving this condition in practice seems easy provided that the number $N$ of measurements is large enough compared to the dimension $n$ of $X$. 

\begin{definition}[self-decomposability amplitude]
Let $X\in \Re^{n\times N}$ be a self-decomposable matrix. We call \emph{self-decomposability amplitude} of $X$, the number $\xi(X)$ defined by 
\begin{equation}\label{eq:xi(X)}
	\xi(X)=\max_{k\in \mathbb{I}}\min_{\gamma_k\in \Re^{N-1}}\Big\{\left\|\gamma_k\right\|_\infty : x_k=X_{\neq k}\gamma_k\Big\}.   
\end{equation}
\end{definition}
The so-defined $\xi(X)$ constitutes a quantitative measure of richness (or genericity) of the regressor matrix $X$. By richness it is meant here how much, in a global sense, the columns of $X$ are linearly independent. $\xi(X)$ is expected to be small if the columns of $X$ are somehow strongly linearly independent.  
\begin{rem}\label{rem:normalization}
If for some $k$ the norm of $x_k$ was to be considerably large in comparison to the norm of the other columns of $X$, then $\xi(X)$ would get large hence reducing recoverability capacity of the considered class of estimators (see also Eq. \eqref{eq:(ii)}). Such situations can be alleviated by normalizing each column of $X$, i.e., for example by replacing $(y_k,x_k)$ by  $(\tilde{y}_k,\tilde{x}_k)\triangleq(y_k/\left\|x_k\right\|,x_k/\left\|x_k\right\|)$ under the assumption that $x_k\neq 0$ for all $k\in \mathbb{I}$. 
\end{rem}

With the help of the device of self-decomposability amplitude \eqref{eq:xi(X)}, we can state a condition for exact recovery of the parameter matrix $A^o$ by solving the optimization problem in \eqref{eq:Psi-YX}.  A similar result was proven in \cite{Bako16-Automatica} for the Least Absolute Deviation (LAD) estimator. 

	\begin{thm}[A sufficient condition for exact recovery]
\label{thm:sufficient-condition-uniqueness}
Let $\varphi$ be a  function satisfying  \eqref{eq:column-summability}-\eqref{eq:Psi-Ineq2} with $\varepsilon=0$ and $\Psi$ be defined as in \eqref{eq:Psi-YX}.  
Assume that $X$ is self-decomposable. 
Then the following statement is true: 
\begin{equation}\label{eq:sufficient-condition-uniqueness}
	\begin{aligned}
	&\forall A\in \Re^{m\times n},\: \forall Y\in \Re^{m\times N}, \\
	&		\left|\mathbb{I}^c(Y-A X)\right|  < T\big(\xi(X)\big) \quad 
						 \Rightarrow \quad  \Psi(Y,X)=\big\{A\big\}. 
	\end{aligned}
\end{equation}
where $T:\Re_{>0}\rightarrow \Re_{>0}$ is the function defined by  $T(\alpha)=	\dfrac{1}{2}\big(1+\dfrac{1}{\alpha}\big)$. 
\end{thm}
\begin{proof}
The proof is completely parallel to that of Theorem 11 in \cite{Bako16-Automatica}. 
From the assumptions, each $x_k$, $k\in \mathbb{I}$, can be written as a linear combination of the columns of $X_{\neq k}$. Let $\gamma_k\in \Re^{N-1}$ be any vector satisfying  
$x_k= X_{\neq k}\gamma_k$. It follows that for any $\Lambda\in \Re^{m\times n}$,
$$\varphi(\Lambda x_k) = \varphi\big(\sum_{t\in \mathbb{I}\setminus \left\{k\right\}}\gamma_{k,t}\Lambda x_t\big)$$
with $\gamma_{k,t}$ denoting the entry of $\gamma_k\in \Re^{N-1}$ indexed by $t$. 
Under the assumptions of the theorem, $\varphi$ is a norm. So, it is positive and satisfies the triangle inequality property. As a result we can write 
$$\varphi(\Lambda x_k) \leq \sum_{t\neq k}\left|\gamma_{k,t}\right| \varphi(\Lambda x_t) \leq \left\|\gamma_k\right\|_\infty\left(\varphi(\Lambda X)-\varphi(\Lambda x_k)\right)$$
where the rightmost term follows from the property \eqref{eq:column-summability} of $\varphi$.
Since this holds for any $\gamma_k$ such that $x_k= X_{\neq k}\gamma_k$, it holds also for 
$$\gamma_k^\star=\argmin_{\gamma\in \Re^{N-1}}\Big\{\left\|\gamma\right\|_\infty : x_k=X_{\neq k}\gamma\Big\}.$$
Hence, 
\begin{equation}\label{eq:sup-inequality}
	\varphi(\Lambda x_k) \leq \xi(X)\left(\varphi(\Lambda X)-\varphi(\Lambda x_k)\right)\: \:  \forall k\in \mathbb{I}, \forall \Lambda \in \Re^{m\times n}.
\end{equation}
or equivalently,
$$\varphi(\Lambda x_k) \leq \dfrac{\xi(X)} {1+\xi(X)}\varphi(\Lambda X)\quad \forall k\in \mathbb{I}, \: \forall \Lambda \in \Re^{m\times n}.$$
 Let $I^c$ be any subset  of $\mathbb{I}$ and pose $\left|I^c\right|=d$.  Summing the previous inequality over the set $I^c$ yields
\begin{equation}\label{eq:concentration-ratio}
	\max_{\Lambda\neq 0}\dfrac{\varphi(\Lambda X_{I^c})}{\varphi(\Lambda X)}\leq \dfrac{1} {2T\big(\xi(X)\big)}\left|I^c\right|
\end{equation}
Note that the term on the right hand side is well-defined since by the self-decomposability assumption, $\rank(X)=n$ which implies that  $\varphi(\Lambda X)\neq 0 $ whenever $\Lambda\neq 0$. 
Therefore \eqref{eq:(ii)}  holds if  $\left|I^c\right|<T\big(\xi(X)\big)$ and the conclusion follows from Theorem \ref{cor:NS-Conds}. 
\end{proof}
It is worth noting that the threshold $T(\xi(X))$ on the number of correctable outliers does not depend on $\varphi$. Hence this threshold is valid  when the estimator is defined from any matrix norm obeying \eqref{eq:column-summability}. 

\begin{rem}
The statement of Theorem \ref{thm:sufficient-condition-uniqueness} still holds true if we replace $\xi(X)$ with the $\varphi$-dependent number $\delta_\varphi(X)$  defined by 
\begin{equation}\label{eq:Lp(X)}
	\delta_\varphi(X)=\max_{k\in \mathbb{I}}\:  \sup_{\Lambda \neq 0}\dfrac{\varphi(\Lambda x_k)}{\varphi(\Lambda X_{\neq k})}
\end{equation}
when it is assumed that $\varphi$ is  a norm and $\rank(X_{\neq k})=n$ for all $k$. 
Doing so will give a less conservative condition for exact recovery. 
However $\delta_\varphi(X)$ seems much harder to evaluate numerically than $\xi(X)$. 
\end{rem}
\vspace{-5pt}
\begin{rem}[A few useful properties of $\xi(X)$] $\: $\hspace{100pt}
\begin{itemize}
	\item For any  nonsingular matrix $R\in \Re^{n\times n}$,  $\xi(RX)=\xi(X)$. It follows that the number $\xi(X)$ 
depends only on the subspace spanned by the rows of the regressor matrix $X$. 
\item For any self-decomposable $X\in \Re^{n\times N}$, $\xi(X)$ is lower-bounded in the following sense 
$$\xi(X)\geq \dfrac{1}{N-1}, $$ 
This follows from the more general observation that 
$$\xi(X)\geq \max_{k\in \mathbb{I}}\dfrac{\left\|x_k\right\|}{\sum_{t\neq k}\left\|x_t\right\|}$$ 
for any vector norm $\left\|\cdot\right\|$. 
As a result, $T(\xi(X)$ is upper-bounded as follows
$$T(\xi(X))\leq \dfrac{N}{2}. $$ 
\end{itemize}
\end{rem}
Theorem \ref{thm:sufficient-condition-uniqueness} provides a sufficient condition for exact recovery in the situation where the function $\varphi$ is a norm. Next, another condition is stated which holds in the general case. 
\begin{prop}
Consider a triplet $(\varphi,\ell,\varepsilon)$ satisfying \eqref{eq:column-summability}-\eqref{eq:Psi-Ineq2}. For  $A\in \Re^{m\times n}$ and $Y\in \Re^{m\times N}$, pose $I^c = \mathbb{I}^c(Y-A X)$, $I^0=\mathbb{I}\setminus I^c=\left\{t\in \mathbb{I}:y_t-Ax_t= 0\right\}$ and $I_\varepsilon^c(\Lambda X_{I^0})=\left\{t\in I^0: \ell(\Lambda x_t)>\varepsilon\right\}$.   
Then $\Psi(Y,X)=\left\{A\right\}$ if  
\begin{equation}\label{eq:SC-Not-Norm}
	\left|I_\varepsilon^c(\Lambda X_{I^0})\right|\varepsilon<\ell(\Lambda X_{I^0})-\ell(\Lambda X_{I^c})
\end{equation}
$\forall \Lambda\in \Re^{m\times n}, \Lambda\neq 0$.
\end{prop}
\begin{proof}
$\Psi(Y,X)=\left\{A\right\}$ is equivalent to
$$\varphi(Y-AX)< \varphi(Y-(A+\Lambda)X)$$
for any $\Lambda\in \Re^{m\times n}$, $\Lambda\neq 0$. Using the definitions of the sets $I^0$ and $I^c$ and applying property \eqref{eq:column-summability} of $\varphi$ yields the equivalent relation
$$\varphi(Y_{I^c}-AX_{I^c})- \varphi(Y_{I^c}-(A+\Lambda)X_{I^c})<\varphi(\Lambda X_{I^0}).$$
By \eqref{eq:Psi-Ineq1}, we can note that $\varphi(Y_{I^c}-AX_{I^c})- \varphi(Y_{I^c}-(A+\Lambda)X_{I^c})\leq \ell(\Lambda X_{I^c})$. It  then follows that 
$$\ell(\Lambda X_{I^c})<\varphi(\Lambda X_{I^0})$$
 is a sufficient condition for $\Psi(Y,X)=\left\{A\right\}$. Finally, invoking  \eqref{eq:Psi-Ineq2} allows us to observe that $\ell(\Lambda X_{I^0})-\left|I_\varepsilon^c(\Lambda X_{I^0})\right|\varepsilon\leq \varphi(\Lambda X_{I^0})$ which implies that $\ell(\Lambda X_{I^c})<\ell(\Lambda X_{I^0})-\left|I_\varepsilon^c(\Lambda X_{I^0})\right|\varepsilon$ is a sufficient condition for $\Psi(Y,X)=\left\{A\right\}$. We have hence proved the proposition. 
\end{proof}

\subsection{Uncertainty set induced by dense noise}
When both $E$ and $F$ are nonzero in the data-generating system \eqref{eq:model-for-robust-estimation}, $\Psi(Y,X)$ is likely to be a non-singleton subset of $\mathscr{P}(\Re^{m\times n})$ especially if we consider all possible realizations of the unknown components $E$ and $F$. In this case the desirable properties of the estimator are in default of better (i) that it contains $A^o$ and (ii) that its size with respect to some metric is as small as possible. 
In this section we are interested in estimating the size of $\Psi(Y,X)$ when both dense noise $E$ and sparse noise $F$ are active in the data-generating system \eqref{eq:model-for-robust-estimation}. 

\paragraph{A notion of  estimator gain}
Similarly to the concept of system gain in control \cite{Zames81-TAC}, one could define the gain of an estimator, that is, a quantitative measure of the sensitivity of the estimator with respect to the perturbations affecting the measurements.  Consider a data pair $(Y,X)$ generated by a system of the form  \eqref{eq:model-for-robust-estimation} with $A^o$ being the parameter matrix sought for. Let us fix the sparse noise matrix $F$ or view it somehow as part of the data-generating system. This consideration proceeds from the fact that $\Psi$ can be insensitive to $F$ (when acting alone)  under, for example, the condition derived in Theorem \ref{thm:sufficient-condition-uniqueness}. Let $E$ be bounded in the sense that $\ell(E)$ is finite with $\ell$ being the norm appearing in \eqref{eq:Psi-Ineq2}. 
Then we can define a gain of the estimator with respect to the dense noise component $E$. 
More specifically, an $(\ell,q)$-gain of the estimator $\Psi$ 
with respect to the dense noise $E$ may be defined by 
\begin{equation}
	g_{\ell,q}(Y,X)=\sup_{\substack{A^\star \in \Psi(Y,X)\\ 0<\ell(E)<\infty\\F\mbox{ \scriptsize sparse}}}\dfrac{\left\|A^\star-A^o\right\|_q}{\ell(E)}.
\end{equation}
Here $\left\|\cdot\right\|_q$ denotes matrix $q$-norm. The so-defined number $g_{\ell,q}(Y,X)$ provides an upper bound on the distance from the set $\Psi(Y,X)$ to $A^o$ in function of the amount of dense noise.  
The following theorem and its corollaries show that if the number of nonzero columns in $F$ is no larger than a certain threshold, then $g_{\ell,q}(Y,X)$ exists and is finite. 
\begin{thm}\label{lem:Bound-Natural}
Let $(Y,X)$ be the data generated by system \eqref{eq:model-for-robust-estimation} subject to the noise components $E$ and $F$.  Consider a triplet $(\varphi,\ell,\varepsilon)$ satisfying \eqref{eq:column-summability}-\eqref{eq:Psi-Ineq2}. Let $S^0\subset \mathbb{I}$ be a set such that $F_{S^0}=0$ and let $S^c=\mathbb{I}\setminus S^0$. 
Assume that the matrix $X$ and the partition $(S^0,S^c)$  are such that there exists $\alpha>0$ such that
 \begin{equation}\label{eq:existence-infimum}
	 \ell(\Lambda X_{S^0})-\ell(\Lambda X_{S^c})\geq \alpha \left\|\Lambda\right\|_q \: \forall \Lambda\in \Re^{m\times n}, 
 \end{equation}
with $\left\|\cdot\right\|_q$ denoting some matrix $q$-norm. \\  
 Then for any $A^\star\in \Psi(Y,X)$, it holds that  
\begin{equation}\label{eq:Error-Bound-L1}
	\left\|A^\star-A^o\right\|_q\leq \dfrac{1}{\gamma_{\ell,q}(X,{S^c})} \big[2\ell(E_{S^0})+\left|I_\varepsilon^c\right|\varepsilon\big]
\end{equation}
with\footnote{The notation $I_\varepsilon^c$ is used for simplicity reasons.} $I_\varepsilon^c =I_\varepsilon^c(Y_{S^0}-A^\star X_{S^0})= \left\{t\in S^0: \ell(y_t-A^\star x_t)> \varepsilon\right\}$ and 
\begin{equation}\label{eq:gamma}
	\gamma_{\ell,q}(X,{S^c})=\inf_{\substack{\Lambda\in \Re^{m\times n}\\\Lambda\neq 0}}\dfrac{\ell(\Lambda X_{S^0})-\ell(\Lambda X_{S^c})}{\left\|\Lambda\right\|_q}
\end{equation}
where $\left\|\cdot\right\|_q$ refers to matrix $q$-norm. 
\end{thm}
\begin{proof}
By definition of $\Psi(Y,X)$ in \eqref{eq:Psi-YX}, 
$$\varphi(Y-A^\star X )\leq \varphi(Y-A X)\:  \forall  A \in \Re^{m\times n} $$
By letting $\Lambda=A-A^o$, $\Lambda^\star=A^\star-A^o$ and applying \eqref{eq:model}, the last inequality takes the form
$$\varphi(F+E-\Lambda^\star X)\leq \varphi(F+E-\Lambda  X)\:  \forall  \Lambda \in \Re^{m\times n}.$$
In particular, for $\Lambda=0$, we get 
$\varphi(F+E-\Lambda^\star X)\leq \varphi(F+E) $ which, thanks to property \eqref{eq:column-summability} of $\varphi$, takes the form
$$\begin{aligned}
	\varphi(F_{S^c}+E_{S^c}-\Lambda^\star X_{S^c})+&\varphi(E_{S^0}-\Lambda^\star X_{S^0})\\
	&\leq \varphi(F_{S^c}+E_{S^c}) +\varphi(E_{S^0}).
\end{aligned}$$
Now applying property \eqref{eq:Psi-Ineq1} to the first member of the left hand side and rearranging yields 
$$\varphi(E_{S^0}-\Lambda^\star X_{S^0})-\ell(\Lambda^\star X_{S^c})\leq \varphi(E_{S^0}).$$

\noindent Using \eqref{eq:Psi-Ineq2} gives 
$$\ell(E_{S^0}-\Lambda^\star X_{S^0})-\left|I_\varepsilon^c\right|\varepsilon-\ell(\Lambda^\star X_{S^c})\leq \varphi(E_{S^0})\leq \ell(E_{S^0}).$$
Here we used the fact that $I_{\varepsilon}^c(E_{S^0}-\Lambda^\star X_{S^0})$ is equal to the set $I_{\varepsilon}^c$ defined in the statement of the theorem. \\
Applying the triangle inequality property of $\ell$, it can be seen that $\ell(\Lambda^\star X_{S^0})-\ell(E_{S^0})\leq \ell(E_{S^0}-\Lambda^\star X_{S^0})$. Combining with the previous inequality yields
$$\ell(\Lambda^\star X_{S^0})-\ell(\Lambda^\star X_{S^c})\leq 2\ell(E_{S^0})+\left|I_\varepsilon^c\right|\varepsilon.$$
Finally, it follows from the definition of $\gamma_{\ell,q}(X,{S^c})$ in \eqref{eq:gamma} that 
$$\gamma_{\ell,q}(X,{S^c})\left\|\Lambda^\star\right\|_q \leq \big[2\ell(E_{S^0})+\left|I_\varepsilon^c\right|\varepsilon\big].$$
The condition \eqref{eq:existence-infimum} guarantees that $\gamma_{\ell,q}(X,{S^c})$ is well-defined and is positive. Hence the statement of the theorem is established. 
\end{proof}
Theorem \ref{lem:Bound-Natural} constitutes an interesting stability result in that it provides a finite upper bound on the distance from  $A^o$ to the set $\Psi(Y,X)$  as a function of the amplitude of the dense noise matrix $E$. It applies to any estimator $\Psi$ defined as in  \eqref{eq:Psi-YX} with $\varphi$ a function  obeying  \eqref{eq:column-summability}-\eqref{eq:Psi-Ineq2}. In particular, in the situation where  $\varphi$ is a norm (in which case $\varepsilon$ can be taken equal to zero in \eqref{eq:Psi-Ineq2}), the inequality in \eqref{eq:Error-Bound-L1} simplifies to
\begin{equation}\label{eq:simplification-estimate}
		\left\|A^\star-A^o\right\|_q\leq \dfrac{2}{\gamma_{\ell,q}(X,{S^c})} \ell(E_{S^0}).
\end{equation}
If $\varphi$ is defined as in \eqref{eq:phi-example} (which, recall, is not a norm) and if the dense noise matrix $E$ is such that $\ell^o(e_t)\leq \varepsilon^o$ for all $t\in \mathbb{I}$, then by taking $\varepsilon=\varepsilon^o$ the set $I_\varepsilon^c$ defined in the statement of Theorem \ref{lem:Bound-Natural} corresponds to the empty set so that \eqref{eq:simplification-estimate} holds as well in this case. In connection with the concept of estimator gain discussed earlier, one can interpret  the factor  $2/\gamma_{\ell,q}(X,{S^c})$ as an estimate of the gain (of the estimator $\Psi$)  with respect to dense noise. 

Lastly, it is interesting to see that when $\varphi$ is a norm, if  $E= 0$ then the result of Theorem  \ref{lem:Bound-Natural} implies that  $\Psi(Y,X)=\left\{A^o\right\}$ provided \eqref{eq:existence-infimum} is true.  

\section{Discussions on some special cases}\label{sec:special-case}
For the purpose of illustrating the extent of the results above, let us discuss further the situation where $\varphi$ reduces to a norm. 

\subsection{Scenario when the loss function is a norm}
\begin{cor}\label{thm:stability-Psi-LossFunction}
Let $(Y,X)$ be the data generated by system \eqref{eq:model-for-robust-estimation} subject to the noise components $E$ and $F$.  
Let $S^0$ and  $S^c$ be defined as in the statement of Theorem \ref{lem:Bound-Natural}.  Assume that  $\varphi$ is a norm i.e., it satisfies  \eqref{eq:column-summability}-\eqref{eq:Psi-Ineq2} with $\varepsilon=0$.  \\
If $X$ is self-decomposable  and $\left|S^c\right|<T\big(\xi(X)\big)$,  then for any $A^\star\in \Psi(Y,X)$, 
\begin{equation}\label{eq:Error-Bound-Gain-Psi}
	\left\|A^\star-A^o\right\|_q\leq \mathscr{B}_{\varphi,q}(|S^0|,X) \varphi(E_{S^0})
\end{equation}
where 
\begin{align}
	&\mathscr{B}_{\varphi,q}(r,X)=\dfrac{2}{\sigma_{\varphi,q}(X)\Big[1-\dfrac{N-r}{ T(\xi(X))}\Big]},\label{eq:Bound} \\
	&\sigma_{\varphi,q}(X)=\inf_{\Lambda \neq 0}\dfrac{\varphi(\Lambda X)}{\left\|\Lambda\right\|_{q}} \label{eq:sigma(X)}
\end{align}
\end{cor}
\begin{proof}
 The principle of the proof is to show that $\gamma_{\ell,q}(X,{S^c})$ is well-defined and then  find a positive underestimate of it. Using the property \eqref{eq:column-summability} of $\varphi$ and the fact that $\varphi=\ell$, we can  write 
$$
\dfrac{\ell(\Lambda X_{S^0})-\ell(\Lambda X_{S^c})}{\left\|\Lambda\right\|_q}= \dfrac{2\varphi(\Lambda X)}{\left\|\Lambda\right\|_q}\left[\dfrac{1}{2}-\dfrac{\varphi(\Lambda X_{S^c})}{\varphi(\Lambda X)}\right]. 
$$
On the other hand we know from the proof of Theorem \ref{thm:sufficient-condition-uniqueness} (see Eq. \eqref{eq:concentration-ratio}) that 
$$\dfrac{\varphi(\Lambda X_{S^c})}{\varphi(\Lambda X)}\leq \dfrac{1}{2T(\xi(X))}\left|S^c\right| $$
so that 
	$$\left[1-\dfrac{\left|S^c\right|}{T(\xi(X))}\right] \dfrac{\varphi(\Lambda X)}{\left\|\Lambda\right\|_q}\leq \dfrac{\ell(\Lambda X_{S^0})-\ell(\Lambda X_{S^c})}{\left\|\Lambda\right\|_q} $$
Taking now the infimum on both sides of the inequality symbol over all nonzero matrices $\Lambda\in \Re^{m\times n}$ yields
$$\sigma_{\varphi,q}(X) \left[1-\dfrac{\left|S^c\right|}{T(\xi(X))}\right]\leq \gamma_{\ell,q}(X,{S^c}).$$	
It follows from the rank condition imposed on $X$ (by the self-decomposability assumption)  that $\sigma_{\varphi,q}(X)>0$. This shows that $\gamma_{\ell,q}(X,{S^c})$ is well defined and is strictly positive.  Finally, since $\varphi=\ell$, invoking  \eqref{eq:simplification-estimate}  gives the result. 
\end{proof}
\noindent Two important comments can be made at this stage. 
\begin{itemize}
	\item First it is interesting to note that the  bound $\mathscr{B}_{\varphi,q}(r,X)$ is an increasing function of $\xi(X)$. Therefore it is all the smaller as  $\xi(X)$ is small. That is, the error bound will be small if the data matrix $X$ is rich enough.  
\item Second,  $\mathscr{B}_{\varphi,q}(r,X)$ is a decreasing function of $r$. This means that the upper bound on the estimation error decreases when the number of gross error columns in $F$ decreases. In the extreme case where $\left|S^0\right|=N$ (no gross error), $\mathscr{B}_{\varphi,q}(|S^0|,X)$ in \eqref{eq:Error-Bound-Gain-Psi} reduces to $2/\sigma_{\varphi,q}(X)$. 
\end{itemize}
Beyond these observations it should be noted that a key assumption of Corollary \ref{thm:stability-Psi-LossFunction} is that  $\left|S^c\right|<T\big(\xi(X)\big)$ with $S^c$ being the index set of the nonzero columns in $F$. Realizing this condition requires on the one hand that the number of nonzero columns in the sparse noise matrix $F$ be small and on the other hand that $\xi(X)$ be small\footnote{Recall that $T$ is a decreasing function hence implying that $T(\xi(X))$ is large when $\xi(X)$ is small.} (which means that the data must be generic). Indeed this condition is not necessarily as strong as it might appear to be at first sight. For example, it can be relaxed as follows.  Observe that the sum $E+F$ is not uniquely defined from model \eqref{eq:model}. Taking advantage of this, one can always absorb in $E$ all nonzero columns of $F$ whose magnitude does not exceed a certain level. To see this, let $I=\left\{t\in S^c: \ell(e_t+f_t)\leq \varepsilon^o\right\}$ where $\varepsilon^o=\max_{t\in \mathbb{I}}\ell(e_t)$. Then we can define $\tilde{E}$  and $\tilde{F}$ such that $E+F=\tilde{E}+\tilde{F}$ and $\tilde{F}_{S^0\cup I}=0$ that is, we set $\tilde{e}_t=f_t+e_t$ and $\tilde{f}_t=0$ for any $t\in I$ and $(\tilde{e}_t,\tilde{f}_t)=(e_t,f_t)$ otherwise. As a consequence, $E$ and $F$ in Corollary \ref{thm:stability-Psi-LossFunction} can be replaced  by $\tilde{E}$ and $\tilde{F}$ respectively so that $|S|$ and $|S^c|$ are replaced by $|S|+|I|$ and $|S^c|-|I|$. The condition of the corollary then becomes  $\left|S^c\right|-\left|I\right|<T\big(\xi(X)\big)$, which is potentially easier to fulfill. 

\begin{rem}[sum of $p$-norms]\label{rem:bound-evaluation}
Evaluating numerically the bound $\mathscr{B}_\varphi(r,X)$ might prove to be a hard problem due to the potential difficulty in computing the term $\sigma_{\varphi,q}(X)$ in \eqref{eq:sigma(X)}. 
 A particular case of interest is when $\varphi$ consists of a sum of $p$-norms of the column vectors, i.e. when it is defined by 
$\varphi(B)=\sum_{i=1}^N\left\|b_i\right\|_{p}$ for $B=\bbm b_1 & \cdots & b_N\eem$. 
In this case if we take $q=2$ in \eqref{eq:Error-Bound-Gain-Psi} and \eqref{eq:sigma(X)}, it is easy to see that $\lambda^{1/2}_{\min}(XX^\top)\leq \sigma_{\varphi,2}(X)$ with $\lambda^{1/2}_{\min}(\cdot)$ denoting the square root of the minimum eigenvalue. Replacing  $\sigma_{\varphi,2}(X)$ with $\lambda^{1/2}_{\min}(XX^\top)$ in \eqref{eq:Bound} yields an  overestimate of $\mathscr{B}_\varphi(r,X)$ which is computable.   
\end{rem}

\begin{rem}
Corollary \ref{thm:stability-Psi-LossFunction} still holds true if one replaces $T(\xi(X))$ with $\pi_{\varphi}^c(X)$ defined in \eqref{eq:pic}. 
As shown in \cite{Sharon09-ACC}, the number $\pi_{\varphi}^c(X)$ in \eqref{eq:pic} is computable although at the price of a combinatorial complexity. However if the $n$-dimension of $X$ is small enough the complexity of the algorithm proposed there can be affordable. Then by using our formula \eqref{eq:Bound} and Remark \ref{rem:bound-evaluation} above, it is possible  therefore to obtain a smaller bound on the estimation error. 
\end{rem}

\subsection{Single output case: $\ell_1$ norm}
In this section, we discuss for an illustrative purpose, the applicability of Theorem \ref{lem:Bound-Natural} to  the case of single-output systems. This is an interesting case to highlight since it represents the most classical situation. Consider the single-output system defined by 
\begin{equation}\label{eq:SISO}
	y_t=(\theta^o)^\top x_t+f_t+e_t
\end{equation}
where $y_t$, $e_t$, $f_t$ are scalars and $x_t$ and $\theta^o$ are $n$-dimensional vectors. By letting $Y=\bbm y_1 & \cdots & y_N\eem \in \Re^{1\times N}$ and defining $E$ and $F$ similarly, we obtain 
\begin{equation}\label{eq:single-output}
	Y=(\theta^o)^\top X+F+E.
\end{equation}
This last equation corresponds indeed to \eqref{eq:model} where the matrix $A^o$ reduces to the row vector $(\theta^o)^\top$. In this case, if we let $\varphi(B)=\sum_{i=1}^N \left\|b_i\right\|_2$ then for any $\theta\in \Re^n$, the columns of (the row vector) $Y-AX$ are scalars so that 
\begin{equation}\label{eq:L1-case}
	\varphi(Y-\theta^\top X)=\sum_{t=1}^N\big\|y_t-\theta^\top x_t\big\|_2=\sum_{t=1}^N\big|y_t-\theta^\top x_t\big|.
\end{equation}
As a result, $\Psi$ coincides in this case with the Least Absolute Deviation (LAD) estimator. 
The following corollary specializes the result of Theorem \ref{lem:Bound-Natural} to the LAD estimator. 
\begin{cor}\label{cor:L1-norm}
Let $(Y,X)\in \Re^{1\times N}\times \Re^{n\times N}$ be generated by model \eqref{eq:SISO}. 
Let $S^c=\big\{t\in \mathbb{I}:f_t\neq 0\big\}$, $S^0=\mathbb{I}\setminus S^c$. 
Assume that $X$ is self-decomposable and $\left|S^c\right|<T\big(\xi(X)\big)$.  
Then for any $\displaystyle\theta^\star \in \argmin_{\theta\in \Re^n}\big\|Y-\theta^\top X\big\|_1$,
$$
	\left\|\theta^\star-\theta^o\right\|_2\leq \mathscr{B}_{1,2}\big(|S^0|,X\big) \left\|E_{S^0}\right\|_1
$$
where 
$$ 
\begin{aligned}
	\mathscr{B}_{1,2}(r,X)&=\dfrac{2}{\sigma_{1,2}(X)\Big[1-\dfrac{N-r}{T(\xi(X))}\Big]},\\
	\sigma_{1,2}(X)&=\inf_{\eta\neq 0}\dfrac{\left\|X^\top \eta\right\|_1}{\left\|\eta\right\|_2}.
\end{aligned}
$$ 
\end{cor}
Again here the bound $\mathscr{B}_{1,2}(r,X)$  can be numerically overestimated by following the idea of Remark \ref{rem:bound-evaluation}. 

\begin{figure*}
\centering
\psfrag{B}[][]{\footnotesize Bound}
\psfrag{alpha}[][]{\footnotesize Percentage of nonzeros [\%]}
\psfrag{pi}{\tiny $\pi_{\varphi}^c(X)$}
\psfrag{T}{\tiny $T(\xi(X))$}
\subfloat[static system: $\xi(X)=0.0083$]{\includegraphics[width=8cm,height=4.5cm]{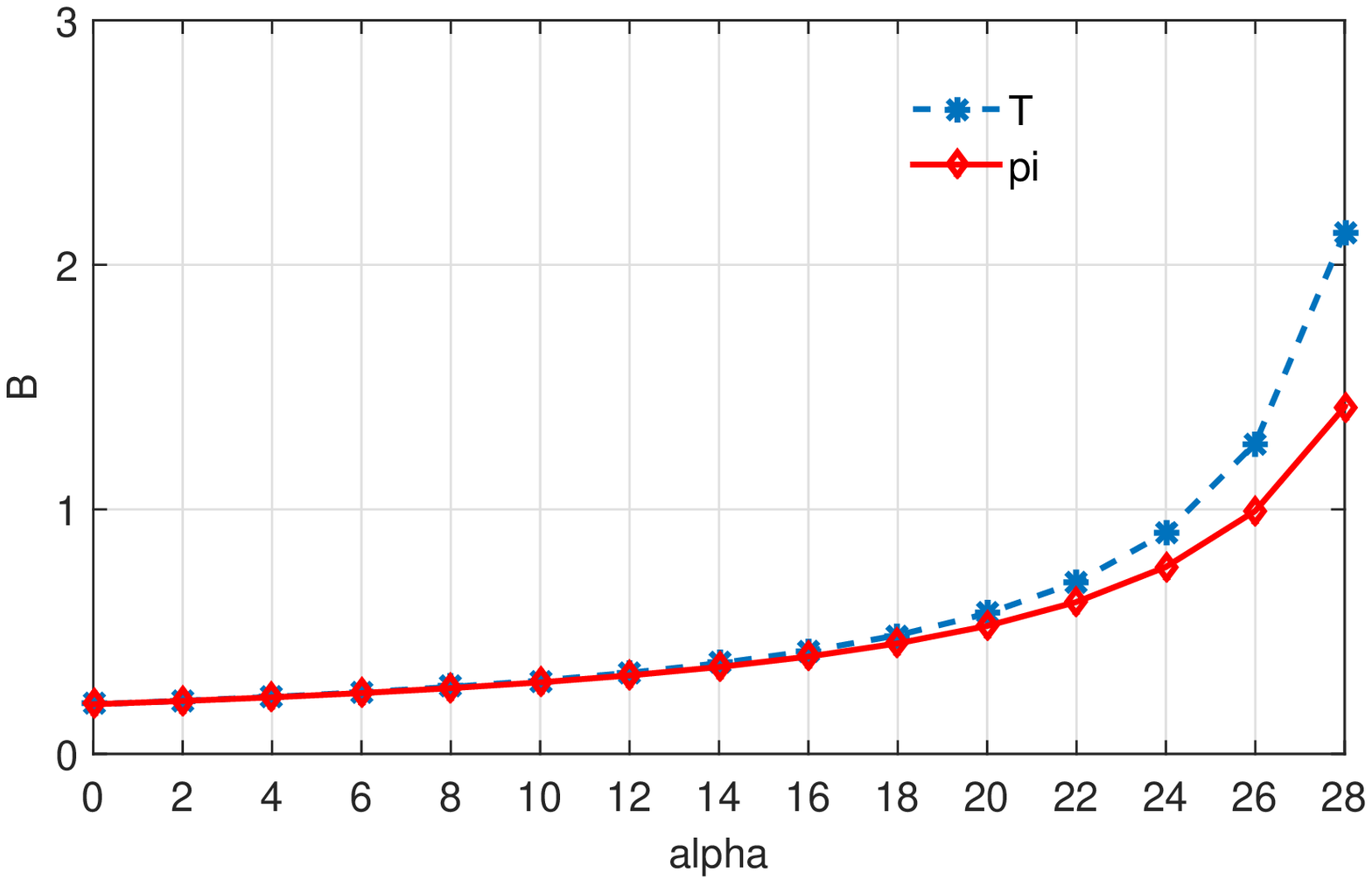}} 
\subfloat[switched system: $\xi(X)=0.0127$]{\includegraphics[width=8cm,height=4.5cm]{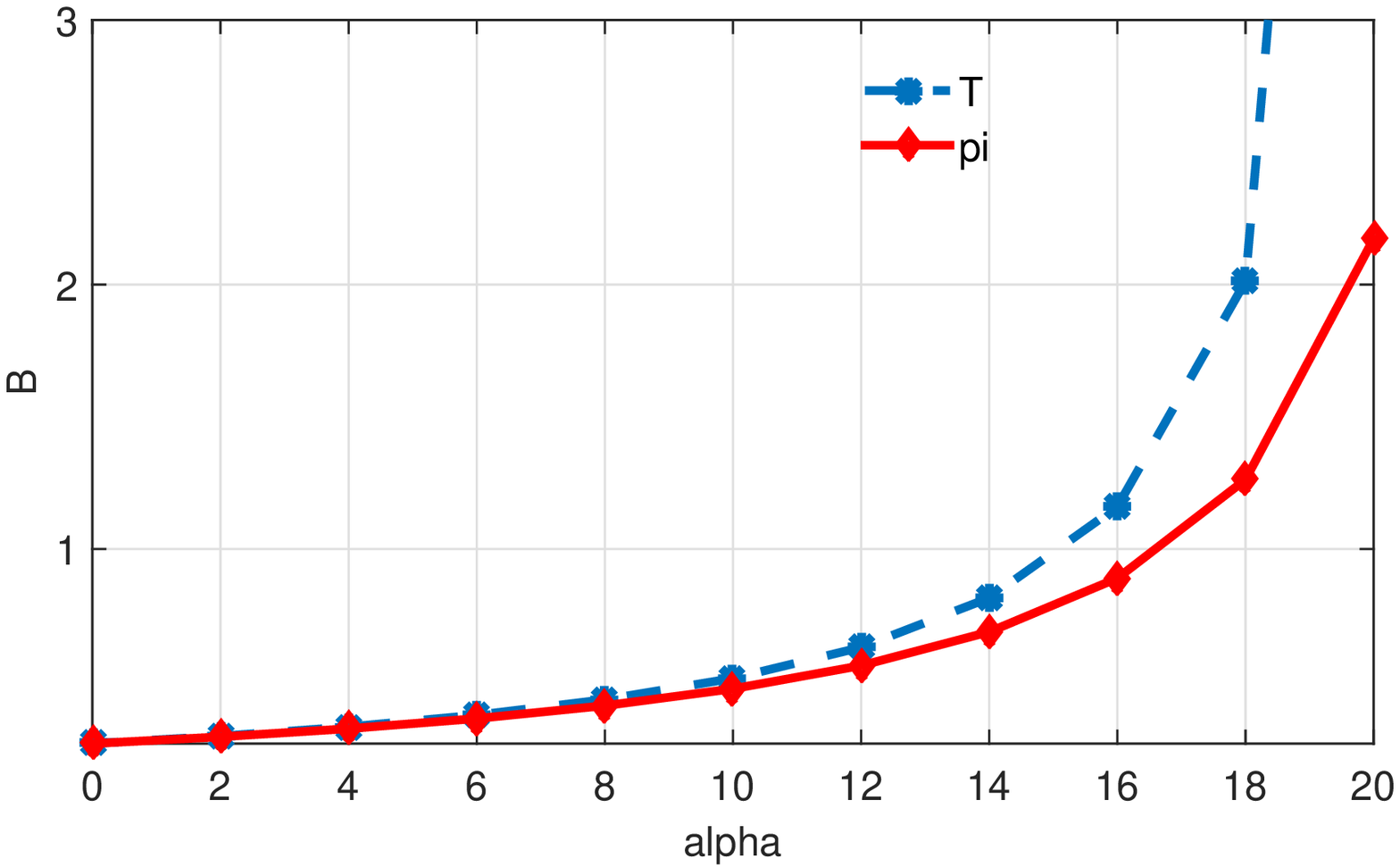}}\\
\subfloat[linear system: $\xi(X)=0.0188$]{\includegraphics[width=8cm,height=4.5cm]{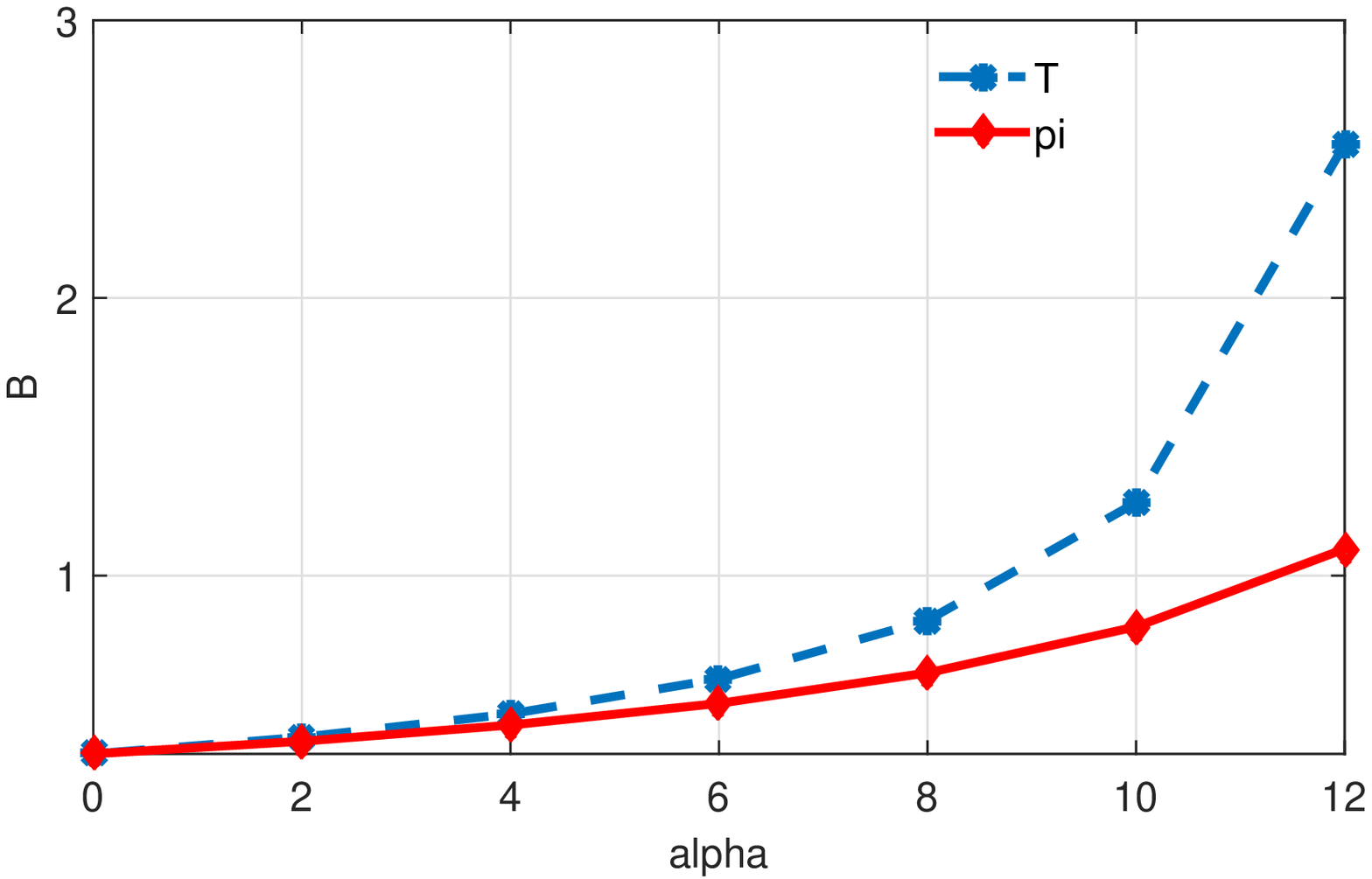}}
\subfloat[nonlinear system: $\xi(X)=0.0107$]{\includegraphics[width=8cm,height=4.5cm]{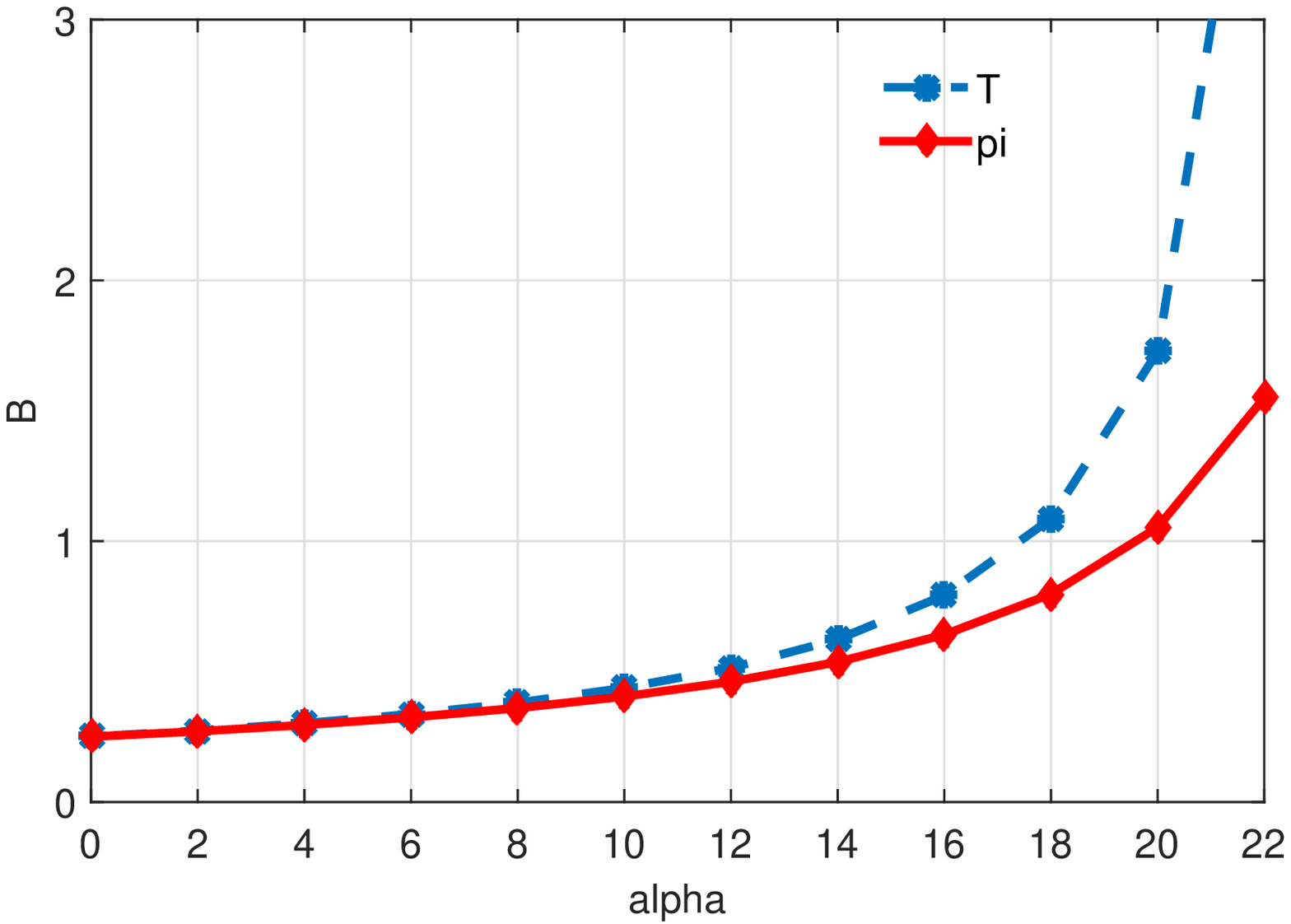}}
\caption{An overestimate of $\mathscr{B}_\varphi$ using respectively $\pi_{\varphi}^c(X)$ and $T(\xi(X))$ for a data matrix $X\in \Re^{2\times 200}$: (a) static data sampled from a Gaussian distribution; (b) data  generated by a switched system; (c)  data  generated by a linear dynamic system ; (d) data  generated by a dynamic nonlinear system. In each case, the x-axis is limited to the range of nonzero gross errors proportions which statisfy the stability condition $\left|S^c\right|/N<T\big(\xi(X)\big)/N$ (see e.g., Corollary \ref{thm:stability-Psi-LossFunction}).}
\label{fig:Bounds}
\end{figure*}
\section{Numerical illustrations}\label{sec:experiments}
The performance of the estimator $\Psi$ has been extensively tested in some existing papers  in the special case of the LAD (see e.g., \cite{Bako16-Automatica}) . We therefore concentrate here on evaluating numerically an estimate of the gain of the estimator based on  Corollary \ref{thm:stability-Psi-LossFunction} and Remark \ref{rem:bound-evaluation}. The estimation is carried out for the case where $\varphi$ consists in the sum of $2$-norms and $q=2$.   Four different cases are studied: 

\begin{itemize}
	\item[(a)]  Static data: $X\in \Re^{2\times 200}$ is sampled from a \textit{Gaussian distribution} $\mathcal{N}(0,I_2)$ with zero-mean and identity-covariance. 
	\item[(b)] Dynamic data generated by a \textit{switched linear system}: $X\in \Re^{2\times 200}$ is formed with the regressors $(y_{t-1}, u_{t-1})$ generated by a switched linear system composed of 3 subsystems of order $1$. 
	This is a switched ARX system defined by $y_t=a_{\sigma(t)}y_{t-1}+b_{\sigma(t)}u_{t-1}$ with the switching signal $\sigma(t)\in \left\{1,2,3\right\}$ generated from a uniform distribution and input $u_t$ being a white noise with Gaussian distribution; $(a_1,b_1)=(-0.40,-0.15)$, $(a_2,b_2)=(1.55,-2.10)$ and $(a_3,b_3)=(1,-0.65)$.  
	\item[(c)]  Dynamic data generated by a \textit{linear ARX system} defined by $y_t=a_1y_{t-1}+b_1u_{t-1}$ with the $(a_1,b_1)$ defined above in case (b).
	\item[(d)]  Dynamic data generated by a \textit{nonlinear NARX system} defined by $y_t=(y_{t-1}+2.5)/(1+y_{t-1}^2)+u_{t-1}$.
\end{itemize}
Following Remark \ref{rem:normalization}, the columns of all data matrices $X$  have been normalized to unit $2$-norm before being processed.

\noindent Figure \ref{fig:Bounds} plots the obtained estimate of the estimator gain  against the proportion of correctable outliers.
As remarked in Section \ref{sec:special-case}, the gain estimate increases as the proportion of outliers gets larger. But the growth rate of the gain estimate depends on the genericity of the data matrix $X$. The more generic the columns of $X$ are, the smaller the growth rate of the estimation error is when regarded as a function of the proportion of outliers. The experiment confirms also the intuition according to which static data tend to be more generic than data generated by a dynamic system. Among the three cases of dynamic systems, the linear system appears to be the one generating the least generic data.


\section{Conclusions}\label{sec:conclusion}
In this paper we have discussed a somewhat general framework for designing a robust estimator. Given the training  data, the estimator is defined as the minimizing set of a certain performance index applying to the data. We have shown that if the performance function possesses some key properties, then the so-defined estimator will inherit robustness properties. Considering a data set generated by a linear model subject to both sparse and dense noises, we showed that the estimator is insensitive to the sparse noise when this latter is acting alone and provided that the number of its nonzero components is no larger than a certain (computable) threshold. Conditions are proposed for the exact recovery of the true parameter matrix when only the sparse noise is active. When both types of noises affect the measurements we propose computable bounds on the parametric estimation error. By assuming stochasticity of the dense noise sequence, the obtained bounds are probably improvable by exploiting appropriately the statistics of the dense noise. This is a matter than can be investigated in future research. 
\section*{Acknowledgement}
The author is grateful to the Associate Editor and the anonymous reviewers for constructive feedback. 

\IEEEtriggercmd{\enlargethispage{-6.5in}}
\IEEEtriggeratref{12}
\bibliographystyle{IEEEtran}

\end{document}